\documentclass[sigconf,nonacm,natbib=false]{acmart}
\AtBeginDocument{%
  }

\setcopyright{acmlicensed}
\copyrightyear{2026}
\acmYear{2026}
\acmDOI{XXXXXXX.XXXXXXX}
\acmConference[VORTEX'26]{}{July 2026}{Brussels}
\acmISBN{978-1-4503-XXXX-X/2026/06}

\RequirePackage[
  datamodel=acmdatamodel,
  style=acmnumeric,
  ]{biblatex}

\begin{document}

\title{A New Syntax and Semantics for Probabilistic Trace Expressions}

\author{Davide Ancona}
\affiliation{%
  \institution{University of Genoa, Department of Informatics, Bioengineering, Robotics and Systems Engineering}
  \city{Genoa}
  \country{Italy}}
\email{davide.ancona@unige.it}

\author{Angelo Ferrando}
\affiliation{%
  \institution{University of Modena and Reggio Emilia, Department of Physics, Informatics and Mathematics}
  \city{Modena}
  \country{Italy}}
\email{angelo.ferrando@unimore.it}

\author{Viviana Mascardi}
\affiliation{%
  \institution{University of Genoa, Department of Informatics, Bioengineering, Robotics and Systems Engineering}
  \city{Genoa}
  \country{Italy}}
\email{viviana.mascardi@unige.it}

\renewcommand{\shortauthors}{D. Ancona, A. Ferrando, and V. Mascardi}

\begin{abstract}
  Runtime Verification (RV) techniques are typically defined under the assumption of complete observability of system executions. In many realistic settings, however, monitors must operate under partial observability, where events may be lost, delayed, or unobservable. This raises fundamental questions about how to interpret specifications, verdicts, and uncertainty during monitoring.
In this paper, we propose a new syntax and semantics for \emph{Probabilistic Trace Expressions} (PTEs), a formal framework that integrates probabilistic reasoning into the operational semantics of Trace Expressions. Trace Expressions (TE) are a highly expressive specification formalism for runtime verification that we started to develop 15 years ago. Rather than attaching probabilities to syntactic transitions, as we did in the original formulation of PTEs dating back 2022, probabilities are now associated with the set of event types enabled in each semantic state, ensuring semantic consistency beyond finite-state models, and high modularity of the PTE specification.
The PTE framework supports principled reasoning about missing events (gaps), distinguishes between observational and generative probabilistic interpretations -- which represents a more refined semantics w.r.t. the original PTE formulation of 2022 --  and subsumes classical probabilistic models such as Hidden Markov Models. We discuss how PTEs enable belief-based monitoring under uncertainty, illustrate their use in one representative Mars Rover scenario, and reflect on the conceptual implications for runtime verification in partially observable environments.

\end{abstract}



%\received{10 April 2026}
%\received[revised]{----}
%\received[accepted]{----}

\maketitle

\newcommand{\emptyseq}{\epsilon}
\newcommand{\Types}{\mathcal{T}}
\newcommand{\Actions}{\mathcal{A}}
\newcommand{\trans}[1]{\stackrel{{#1}}{\rightarrow}}
\newcommand{\transM}[1]{\stackrel{{#1}}{\rightarrow_{\gamma}}}
\newcommand{\transo}[1]{\stackrel{{#1}}{\rightarrow_{o}}}
\newcommand{\transg}[1]{\stackrel{{#1}}{\rightarrow_{g}}}
\newcommand{\transS}[1]{\stackrel{{#1}}{\twoheadrightarrow}}
\newcommand{\ntrans}[1]{\stackrel{{#1}}{\nrightarrow}}
\newcommand{\transmsg}[1]{\stackrel{{#1}}{\Longrightarrow}}
\newcommand{\isEmpty}{\mathit{\epsilon}}
\newcommand{\acts}{\mathit{tr}}
\newcommand{\ping}{\mathit{ping}}
\newcommand{\pong}{\mathit{pong}}
\newcommand{\msgone}{\mathit{msg_1}}
\newcommand{\msgtwo}{\mathit{msg_2}}
\newcommand{\msgthree}{\mathit{msg_3}}
\newcommand{\ackone}{\mathit{ack_1}}
\newcommand{\acktwo}{\mathit{ack_2}}
\newcommand{\ackthree}{\mathit{ack_3}}
\newcommand{\Loop}{\mathit{Forever}}
\newcommand{\AltBitOne}{\mathit{AltBit_1}}
\newcommand{\AltBitTwo}{\mathit{AltBit_2}}
\newcommand{\AltBitThree}{\mathit{AltBit_3}}
\newcommand{\Mone}{\mathit{M_1}}
\newcommand{\Mtwo}{\mathit{M_2}}
\newcommand{\fstOptMone}{\mathit{A\Mone}}
\newcommand{\sndOptMone}{\mathit{B\Mone}}
\newcommand{\fstOptMtwo}{\mathit{A\Mtwo}}
\newcommand{\sndOptMtwo}{\mathit{B\Mtwo}}
\newcommand{\Rule}[4]{\scriptstyle{\textrm{({#1})}}{\displaystyle\frac{#2}{#3}}\ #4}
\newcommand{\satype}{\alpha} %%% sending action type
\newcommand{\sa}{a} %%% sending action

\newcommand{\WAB}{\mathit{WAB}}
\newcommand{\AB}{\mathit{AB}}
\newcommand{\MA}{\mathit{MA}}
\newcommand{\AM}{\mathit{AM}}
\newcommand{\MM}{\mathit{MM}}
\newcommand{\BS}{\mathit{BS}}
\newcommand{\OB}{\mathit{OB}}
\newcommand{\BC}{\mathit{BC}}
\newcommand{\msg}{\mathit{msg}}
\newcommand{\ackType}{\mathit{ack}}
\newcommand{\pmsgone}[1]{\mathit{msg_1^{#1}}}
\newcommand{\pmsgtwo}[1]{\mathit{msg_2^{#1}}}
\newcommand{\packone}[1]{\mathit{ack_1^{#1}}}
\newcommand{\packtwo}[1]{\mathit{ack_2^{#1}}}
\newcommand{\pmsgthree}[1]{\mathit{msg_3^{#1}}}
\newcommand{\packthree}[1]{\mathit{ack_3^{#1}}}
\newcommand{\offer}{\mathit{offer}}
\newcommand{\buy}{\mathit{buy}}
\newcommand{\close}{\mathit{close}}
\newcommand{\closerel}{\triangleright}
\newcommand{\ExtTypes}{\mathcal{T}^+}
\newcommand{\emb}[1]{\mathit{emb}({#1})}
\newcommand{\invemb}[1]{\mathit{emb}^{-1}({#1})}
\newcommand{\extrans}[1]{\stackrel{{#1}}{\rightarrowtail}}
\newcommand{\optrans}[1]{\stackrel{{#1}}{{\rightarrowtail\hspace*{-1em}\longrightarrow}}}
\newcommand{\runset}{\mathcal{R}}
\newcommand{\proj}{\mathit{proj}}
\newcommand{\annot}{\nu}
\newcommand{\pth}{p}
\newcommand{\asa}{b}
\newcommand{\eventSet}{\mathcal{E}}
\newcommand{\eventTy}{\vartheta}
\newcommand{\eventTyProb}[1]{\vartheta[{#1}]}
\newcommand{\anyEvTy}{\mathit{any}}
\newcommand{\noneEvTy}{\mathit{none}}
\newcommand{\ev}{e}
\newcommand{\evs}[1]{\ev_1, ..., \ev_{#1}}
\newcommand{\es}{es}
\newcommand{\aevent}{\mathit{ae}}
\newcommand{\run}{\mathit{run}}
\newcommand{\arun}{\mathit{annot\_run}}
\newcommand{\cover}{C}
\newcommand{\shift}[2]{{#1}\downarrow{#2}}
\newcommand{\State}[2]{\langle {#1}, {#2} \rangle}
\newcommand{\globalState}[3]{\langle {#1}, {#2}, {#3} \rangle}
\newcommand{\gap}[1]{\gapg({#1})}
\newcommand{\gapg}{\mathit{gap}}
\newcommand{\any}[1]{\anyg({#1})}
\newcommand{\anyg}{\mathit{eg}}
\newcommand{\poss}{poss\ET}
\newcommand{\pio}{\pi o}
\newcommand{\pig}{\pi g}
%%% math keywords
\newcommand{\alice}{\mathit{alice}}
\newcommand{\bob}{\mathit{bob}}
\newcommand{\seller}{\mathit{seller}}
\newcommand{\buyer}{\mathit{buyer}}
\newcommand{\tell}{\mathit{tell}}
\newcommand{\price}{\mathit{price}}
\newcommand{\pasta}{\mathit{pasta}}
\newcommand{\pizza}{\mathit{pizza}}
\newcommand{\itemv}{\mathit{item}}
\newcommand{\isItem}{\mathit{isItem}}
\newcommand{\isSeller}{\mathit{isSeller}}
\newcommand{\isBuyer}{\mathit{isBuyer}}
\newcommand{\amount}{\mathit{amount}}
\newcommand{\safe}{\mathit{safe}}
\newcommand{\isEmptyMeth}{\mathit{isEmpty}}
\newcommand{\sizeMeth}{\mathit{size}}
\newcommand{\popMeth}{\mathit{pop}}
\newcommand{\topMeth}{\mathit{top}}
\newcommand{\pushMeth}{\mathit{push}}
\newcommand{\Stack}{\mathit{Stack}}
\newcommand{\NEStack}{\mathit{NEStack}}
\newcommand{\sem}[1]{\llbracket{#1}\rrbracket}
\newcommand{\semThree}[1]{\llbracket{#1}\rrbracket_3}
\newcommand{\Tops}{\mathit{Tops}}

\newcommand{\role}{\mathit{role}}
\newcommand{\hasRole}{\mathit{hasRole}}
\newcommand{\op}{op}
\newcommand{\prefixop}{{:}}
\newcommand{\orop}{{\vee}}
\newcommand{\shuffleop}{{|}}
\newcommand{\catop}{{\cdot}}
\newcommand{\allop}{1}
\newcommand{\emptyop}{0}
\newcommand{\andop}{{\wedge}}
\newcommand{\condop}{{\gg}}
\newcommand{\aorb}{\mathit{a\_or\_b}}
\newcommand{\borc}{\mathit{b\_or\_c}}
\newcommand{\TS}{\mathit{TS}}
\newcommand{\restrict}[2]{{#1}_{\mid{#2}}}
\newcommand{\prefixRel}{\ltimes}
\newcommand{\TE}{\mathit{TE}}

\newtheorem{myexample}{Example}[section]

\newcommand{\Guest}{\mathit{Guest}}
\newcommand{\Access}{\mathit{Access}}
\newcommand{\Admin}{\mathit{Admin}}
\newcommand{\RW}{\mathit{RW}}
\newcommand{\login}{\mathit{login}}
\newcommand{\logout}{\mathit{logout}}
\newcommand{\rread}{\mathit{read}}
\newcommand{\wwrite}{\mathit{write}}

\newcommand{\stack}{\mathit{stack}}
\newcommand{\Any}{\mathit{Any}}
\newcommand{\Unsafe}{\mathit{Unsafe}}
\newcommand{\unsafe}{\mathit{unsafe}}

\newcommand{\msgack}{\mathit{msg\_ack}}
\newcommand{\AltBit}{\mathit{AltBit}}

\newcommand{\AP}{\mathit{AP}}
\newcommand{\true}{\mathit{true}}
\newcommand{\until}[2]{{#1}\,U{#2}}
\newcommand{\Inf}{\mathit{Inf}}
\newcommand{\ET}{\mathcal{ET}}

\newcommand{\Proj}{\texttt{Proj}}
\newcommand{\co}{\texttt{CuriosOne}}
\newcommand{\coo}{\texttt{CuriosOnes}}
\newcommand{\gc}{\texttt{GroundControl}}
\newcommand{\mw}{\texttt{MarsWarehouse}}

\newcommand{\new}[1]{{\color{red}{#1}}}

\newcommand{\viv}[1]{\begin{center}\fbox{\parbox{\textwidth}{\textcolor{blue}{Viviana: {#1}}}}\end{center}}

\section{Introduction}
\label{intro-sec}

The need of developing reliable software is as old as the software itself. 

The notion of reliability of software applications has evolved over time. 
In the mid seventies, Linden defined a reliable software as ``a software that provides services which are (1) usable, (2) correct, (3) trustworthy, and (4) available on demand'' \cite{Linden:1976:OSS:356678.356682}.

Ten years later, Laprie stated that dependability and reliability are closely connected, with dependability being more general than reliability. His definition of dependability is ``the quality of the delivered service such that reliance can justifiably be placed on this service'' \cite{laprie1985dependable}.  He also observed that ``a system failure occurs when the delivered service deviates from the specified service, where the service specification is an agreed description of the expected service''. 

Laprie's definition of software failure has been rephrased as ``a deviation between the observed behaviour and
the required behaviour of a software system'' by Delgado et al. \cite{DBLP:journals/tse/DelgadoGR04}.

The ``study, development, and application of those verification techniques that allow checking whether a run of a system under scrutiny
satisfies or violates a given correctness property is named Runtime Verification'' \cite{DBLP:journals/jlp/LeuckerS09}.

Runtime Verification (RV) of complex, distributed systems under ideal conditions (perfect observability of all the relevant events, no leaky communication channels, etc) is a hard task to perform, and has been addressed by many scientific works including surveys and introductory papers \cite{DBLP:series/lncs/BartocciFFR18,pastExp,DBLP:journals/jlp/LeuckerS09}, books \cite{DBLP:series/lncs/10457}, seminars \cite{DBLP:conf/dagstuhl/2010P10451}, and conferences\footnote{\url{http://www.runtime-verification.org/}, accessed on \today.}.
The IC1402 - Runtime Verification beyond Monitoring (ARVI) COST action represented a significant achievement for the RV community\footnote{ARVI IC1402, supported by COST (European Cooperation in Science and Technology) from 17/12/2014 to 16/12/2018, \url{https://www.cost-arvi.eu/}, accessed on \today.}, having been funded to overcome the fragmentation of RV research by designing common input formats for RV tool cooperation and comparison, evaluating different RV tools, and identifying the RV road-map for the next years.

%%%%%%%%%%%%%%%%%%%%%%%%% GAPS %%%%%%%%%%%%%%%%%%%%%%%%%%%%%
%
When the conditions are not ideal and some events that took place in the system are missed by the monitor, generating a \emph{gap} in the trace of observed events, performing RV becomes even harder.
The problem of partial observability and lossy traces started to be addressed less than ten years ago.
%

%%%%%%%%%%%%%%%%%% decentralised %%%%%%%%%%%%%%%%%%%%%%%%%5
Besides partial observability, another issue affecting the efficiency and reliability of the RV approach is centralisation. If only one monitor is in charge for  observing the whole system under scrutiny, that monitor may soon become a bottleneck.

Distributed RV (DRV) addresses this problem. 
As stated in the manifesto of the First Workshop on DRV held in Bertinoro in May 2016\footnote{\url{http://www.labri.fr/perso/travers/DRV2016/}, accessed on \today.}, ``DRV is about designing a fault-tolerant distributed algorithm that monitors another distributed algorithm, with the end goal of developing lightweight software systems that are more efficient that traditional verification techniques.''
% %
As for partial observability, also RV decentralisation has been faced only recently, and by a few scientists.
% %
In 2012 Bauer and Falcone \cite{DBLP:conf/fm/BauerF12} proposed an algorithm for distributing and monitoring LTL formulae, such that satisfaction or violation of specifications can be detected by decentralised monitors; Bartocci extended that work to take into account how to set a sampling time among the components such that their local traces are consistent \cite{Bartocci_2013}.
% %
The first survey on DRV has been published in 2016 \cite{Bonakdarpour2016}. It references 18 papers only, and many of them deal with issues which fall outside the scope of our investigation, such as evaluating the minimum number of opinions needed for fault-tolerant RV of distributed systems \cite{Fraigniaud2014b,Fraigniaud2014a} and wait-free synchronization \cite{Herlihy:1991:WS:114005.102808}. A more recent survey published in 2018 \cite{DBLP:series/lncs/FrancalanzaPS18} referencing more than 90 papers shows that the awareness of the challenges raised by  DRV approaches has dramatically increased in the last few years. \\

Despite this ever growing attention to DRV, to the best of our knowledge no works addressing {\em both} gaps/lossy traces management and DRV exist. This is where our proposal comes into play. In this paper we define a new, simplified syntax for the Probabilistic Trace Expressions (PTEs) formalism we presented in 2022 \cite{DBLP:conf/eumas/AnconaFM22,DBLP:conf/cilc/AnconaFM22} and we show how it can be exploited for RV of system, even in presence of observation gaps. We explore the formal connections among HMMs, LTL, and PTEs, and we suggest that algorithms to perform RV of the system under scrutiny may be decentralised under suitable conditions.
% \\

While decentralised monitoring strategies can indeed be built on top of PTEs, this paper focuses on the foundational semantics and probabilistic reasoning required for runtime verification under partial observability.

The paper is organised as follows: Section \ref{sec:backgroundRel} introduces the background for PTEs and overviews related works. Section \ref{motiv-sec} presents our motivating scenario inspired by ongoing Mars exploration programs and gently introduces PTEs; Section \ref{prob-trace-expr-sec} formalises the new, simplified PTE syntax and two different semantics, {\em ObSem} and {\em GuesSem}, which represent a brand new contribution with respect to previous preliminary works on the topic \cite{DBLP:conf/eumas/AnconaFM22,DBLP:conf/cilc/AnconaFM22}. Section \ref{sec:theory} shows how HMMs and LTL formulae can be both translated into PTEs and demonstrates that a PTE resulting from a HMM enjoys the same properties as the original HMM according to probability propagation via the forward algorithm;  Section \ref{impl-sec} discusses some practical aspects to take into account when dealing with PTEs; Section \ref{concl-sec} concludes and outlines the future directions of our work.

\section{Background and Related Work}
\label{sec:backgroundRel}

\subsection{Trace Expressions}
\label{subsect:TE}

Trace Expressions (TEs) \cite{DBLP:conf/atal/AnconaBFM15,DBLP:conf/dalt/AnconaDM12}
specify sets of event traces over a fixed universe of events $\eventSet$. Events are abstracted into \emph{event types}, drawn from a set $\ET$, each denoting a subset of $\eventSet$. An event $\ev$ matches an event type $\eventTy$ if $\ev \in \eventTy$.

TEs are generated by the following constructs:
\[
\emptyseq \quad
\eventTy \prefixop \tau \quad
\tau_1 \catop \tau_2 \quad
\tau_1 \andop \tau_2 \quad
\tau_1 \orop \tau_2 \quad
\tau_1 \shuffleop \tau_2
\]
denoting, respectively, the empty trace, prefixing, concatenation, intersection, union, and shuffle.
TEs support recursion via finite systems of syntactic equations.

The denotational semantics $\sem{\tau}$ of a TE $\tau$ is the set of traces it denotes.
Operationally, TEs are equipped with a transition relation
$\tau \trans{\ev} \tau'$, meaning that upon observing event $\ev$, the protocol may evolve from state $\tau$ to $\tau'$.

We define the transitive closure of $\trans{\ev}$ as $\tau\trans{evs}\tau_k$ iff $evs = \ev_1~\ev_2~...~\ev_k$ and there exist $\tau_1, \tau_2, ..., \tau_j, ...$ such that
$\tau\trans{\ev_1}\tau_1\trans{\ev_2}\tau_2 .... \trans{\ev_k}\tau_k$.
%\\
%The transitive closure $\tau \trans{evs} \tau'$ is defined as usual.

%\vspace{-0.5cm}
\begin{figure*}[!h]
\begin{center}
\begin{math}
\begin{array}{c}
\Rule{prefix}
{}
{\eventTy\prefixop\tau\trans{\ev}\tau}
{\ev \in \eventTy}
\qquad
\Rule{or-l}
{\tau_1\trans{\ev}\tau'_1}
{\tau_1\orop\tau_2\trans{\ev}\tau'_1}
{}
\qquad
\Rule{or-r}
{\tau_2\trans{\ev}\tau'_2}
{\tau_1\orop\tau_2\trans{\ev}\tau'_2}
{}

\\[6ex]

\Rule{and}
{\tau_1\trans{\ev}\tau'_1\quad\tau_2\trans{\ev}\tau'_2}
{\tau_1\andop\tau_2\trans{\ev}\tau'_1\andop\tau'_2}
{}
\qquad
\Rule{shuffle-l}
{\tau_1\trans{\ev}\tau'_1}
{\tau_1\shuffleop\tau_2\trans{\ev}\tau'_1\shuffleop\tau_2}
{}
\qquad
\Rule{shuffle-r}
{\tau_2\trans{\ev}\tau'_2}
{\tau_1\shuffleop\tau_2\trans{\ev}\tau_1\shuffleop\tau'_2}
{}

\\[6ex]

\Rule{cat-l}
{\tau_1\trans{\ev}\tau'_1}
{\tau_1\catop\tau_2\trans{\ev}\tau'_1\catop\tau_2}
{}
\qquad
\Rule{cat-r}
{\tau_2\trans{\ev}\tau'_2}
{\tau_1\catop\tau_2\trans{\ev}\tau'_2}
{\isEmpty(\tau_1)}
\\[4ex]
\Rule{$\isEmpty$-empty}
{}
{\isEmpty(\emptyseq)}
{}
\qquad
\Rule{$\isEmpty$-or-l}
{\isEmpty(\tau_1)}
{\isEmpty(\tau_1\orop\tau_2)}
{}
\qquad
 \Rule{$\isEmpty$-or-r}
{\isEmpty(\tau_2)}
{\isEmpty(\tau_1\orop\tau_2)}
{}
\qquad
\Rule{$\isEmpty$-others}
{\isEmpty(\tau_1)\quad\isEmpty(\tau_2)}
{\isEmpty(\tau_1 \op\ \tau_2)}
{\op\in\{\shuffleop,\catop,\andop\}}
\end{array}
\end{math}

\end{center}
\caption{Transition system for trace expressions.}\label{trans-fig}
\end{figure*}
%\vspace{-0.5cm}
We restrict attention to \emph{contractive} TEs, where every infinite path contains a prefix operator, ensuring well-defined behaviour.
A TE is \emph{deterministic} if any two states reached by the same event trace denote the same set of traces.

\subsection{Linear-Time Temporal Logic}
\label{LTL-subsect}
Linear-time Temporal Logic (LTL \cite{focs1977-Pnu}) is widely used in model checking \cite{baier2008principles} and runtime verification (RV). For RV, LTL is commonly equipped with a three-valued semantics \cite{LeuckerBS09Tosem} to account for incomplete traces.
An LTL formula is built from atomic propositions using Boolean connectives and temporal operators \textbf{X} (next) and \textbf{U} (until). Derived operators include \textbf{F} (eventually) and \textbf{G} (always). LTL formulae are evaluated over infinite traces of sets of atomic propositions, with the standard satisfaction relation.

\subsection{Hidden Markov Models}
\label{HMM-subsect}
A Hidden Markov Model (HMM \cite{baum1966,Rabiner86anintroduction}) is defined as a tuple
$H=\langle S,A,V,B,\Pi\rangle$,
where $S$ is the set of hidden states, $A$ the state transition matrix, $V$ the observation alphabet, $B$ the observation probability matrix, and $\Pi$ the initial state distribution.
%
%In Section \ref{motiv-sec} we adapt a Mars rover example inspired by \cite{DBLP:conf/rv/StollerBSGHSZ11}, where commands issued to rover instruments generate observable events such as command dispatch, success, or failure, possibly with loss. 
Observations are naturally grouped into \emph{event types}, a notion implicit in the HMM literature and made explicit in PTEs.
In HMMs, the probability of being in a given state after observing a sequence is computed via the forward algorithm \cite{Rabiner:1990:THM:108235.108253}. This algorithm plays a central role in relating HMMs to PTE semantics.

\subsection{Related Work}
\label{sec:related}

One of the oldest papers dealing with partial observability in RV is ``Runtime Verification with State Estimation'' by Stoller et al. \cite{DBLP:conf/rv/StollerBSGHSZ11}. In that paper, the authors  introduce the concept of RV with State Estimation and apply it to estimate the probability that a temporal property is satisfied by a run of a program when monitoring overhead is
reduced by sampling, generating gaps
in the observed traces. They consider event traces as observation sequences of a Hidden Markov Model (HMM), use a HMM model of the
monitored program to fill in sampling-induced gaps in observation sequences, and extend the forward algorithm for HMM state estimation by Rabiner \cite{Rabiner:1990:THM:108235.108253} to
compute the probability that the property is satisfied by a trace execution.
Basin et al.  present a policy language with a three-valued semantics that supports reasoning about incomplete knowledge and handling disagreements \cite{DBLP:conf/rv/BasinKMZ12}, and Bartocci et al. \cite{DBLP:conf/rv/BartocciGKSSZS12} combine overhead control, RV with state estimation as presented in \cite{DBLP:conf/rv/StollerBSGHSZ11}, and predictive analysis. 
Along this line, Kalajdzic et al. \cite{DBLP:conf/rv/KalajdzicBSSG13} propose a way to control the trade-off between uncertainty and overhead in RV thanks to low-cost observations  of parts of the program state which are performed probabilistically at the end of observation gaps.
More recently, Joshi et al. \cite{DBLP:conf/sac/JoshiTF17} have presented an offline algorithm  that  identifies  whether a Linear time  Temporal
Logic (LTL) formula can be soundly monitored in the presence of a
transient
loss of events in a trace and constructs a monitor accordingly; Babaee et al. \cite{DBLP:conf/sefm/BabaeeGF18} estimate the finite extensions of an event trace prefix using a HMM as predictive model and a Deterministic
Finite Automaton to express the extensions. Depending on the given property, the extensions may specify the prefixes that satisfy the property (good extensions) or violate it (bad extensions). Ferrando et al.~\cite{DBLP:conf/sefm/FerrandoM22} tackle runtime verification under \emph{imperfect information} by making partial observability explicit at the level of atomic propositions. Their key idea is to model what the monitor can and cannot distinguish via \emph{indistinguishability sets} of propositions, and to compile an LTL formula into an \emph{explicit} version whose atoms are tagged with witnesses of truth and falsity. The paper also provides an engineering-oriented monitoring workflow and validates it on a robotic case study. Finally, Cimatti et al. \cite{DBLP:conf/rv/CimattiTT19} generalise the RV framework based on propositional LTL with both future and past temporal operators to monitor partially observable systems. Models of the system under scrutiny are used as assumptions for reasoning on the non-observable or future behaviours of the system.

%%%%%%%%%%%%%%%%%%%%%%%%%%%%%%%%%%%%%%%%%%%%%%%%%%%%%%%%%%%%%%%%%%%%%%%%%%%%%%%

\section{Motivating Scenario: The Lonely Mars \co}
\label{motiv-sec}

\begin{figure*}[!htb]
\begin{center}
\includegraphics[width=0.6\linewidth]{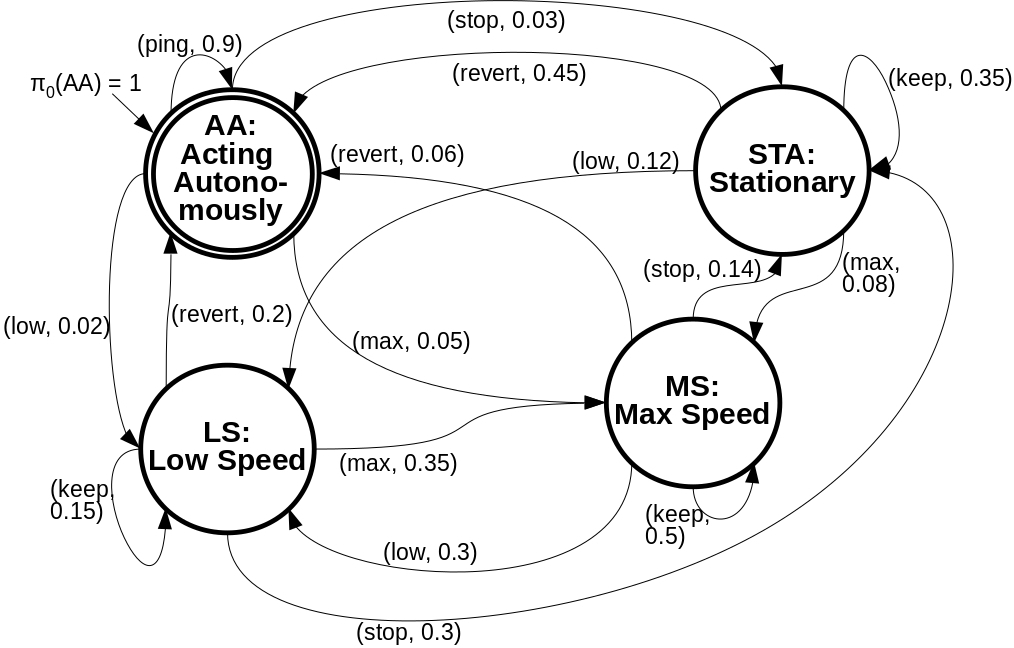}
\end{center}
\caption{The PFSA modelling COGC.}\label{PFSA-fig}
\end{figure*}

Let us consider the following scenario, inspired to the {\em Curiosity} mission.
On August 6, 2012, \co\ safely reaches Mars and begins its autonomous mission.
Although \co\ can operate independently, mission control (\gc) may intervene at any time.
Due to the high communication latency between Earth and Mars and the cost of transmission,
\gc\ follows a strict protocol: every 30 minutes it sends a control message to \co.

Most messages are lightweight \emph{ping}s.
Occasionally, \gc\ sends commands to change speed (\texttt{low} or \texttt{max}),
to stop (\texttt{stop}), or to explicitly return control (\texttt{revert});
a \texttt{keep} message may be used to maintain the current state.
The resulting communication protocol, called \textsc{COGC}, is ongoing and has no terminal state.

\textsc{COGC} can be naturally represented as a Probabilistic Finite State Automaton (PFSA),
whose states correspond to operational modes
(acting autonomously, stopped, low speed, maximum speed),
and whose transitions are labelled by control messages with associated probabilities.
From each state, the probabilities of outgoing transitions sum to~1
(Figure~\ref{PFSA-fig}).

The same protocol can also be represented using Trace Expressions (TEs)~\cite{DBLP:conf/birthday/AnconaFM16},
where protocol states are expressions and transitions are driven by observed communicative events.

The COGC protocol can hence be represented by the AA (for ``Acting Autonomously'') TE depicted in Figure \ref{ex:AA}.
\begin{figure*}[!htb]
\begin{center}
\begin{math}
\begin{array}{c}
AA = (gc \transmsg{ping} co ~\prefixop~ AA) ~\orop~ (gc \transmsg{stop} co ~\prefixop~ STA) ~\orop~ (gc \transmsg{low} co ~\prefixop~ LS) ~\orop~ (gc \transmsg{max} co ~\prefixop~ MS)
\\[3ex]

STA = (gc \transmsg{keep} co ~\prefixop~ STA) ~\orop~ (gc \transmsg{revert} co ~\prefixop~ AA) ~\orop~ (gc \transmsg{low} co ~\prefixop~ LS) ~\orop~ (gc \transmsg{max} co ~\prefixop~ MS)
\\[3ex]

LS = (gc \transmsg{keep} co ~\prefixop~ LS) ~\orop~ (gc \transmsg{revert} co ~\prefixop~ AA) ~\orop~ (gc \transmsg{stop} co ~\prefixop~ STA) ~\orop~ (gc \transmsg{max} co ~\prefixop~ MS)
\\[3ex]

MS = (gc \transmsg{keep} co ~\prefixop~ MS) ~\orop~ (gc \transmsg{revert} co ~\prefixop~ AA) ~\orop~ (gc \transmsg{stop} co ~\prefixop~ STA) ~\orop~ (gc \transmsg{low} co ~\prefixop~ LS)
\\[3ex]
\end{array}
\end{math}
\end{center}
\raggedright This TE has four different states, $AA$, $STA$, $LS$, $MS$, and from each of them four distinct transitions to $AA$, $STA$, $LS$, $MS$ are possible. 
If we want to make explicit that the initial state of $COGC$ is $AA$ we may add the following equation:
$$
COGC = AA
$$
\caption{The $COGC$ protocol represented by the $AA$ (for ``Acting Autonomously'') Trace Expression.}\label{ex:AA}
\end{figure*}

The TE  representation  captures the protocol structure without needing to fix the number of states in advance; also, TEs are more expressive than FSA. For instance, protocols that require matching push and pop events, such as stack-based
buffering, cannot be represented by finite-state automata but are naturally captured
by Trace Expressions.
Nevertheless, TEs cannot represent probabilities, that PFSAs represent and that are needed to deal with partial observability.  

When communication is reliable, observing a message uniquely determines the next protocol state.
However, communication failures are unavoidable.
If \co\ does not receive a message at the expected time, it can only infer that some message
was sent but lost; we say that \co\ observes a \emph{gap}.

At this point, \co\ must reason probabilistically.
Suppose that \co\ has received a sequence of \texttt{ping} messages and remains
in its autonomous state.
If the next message is missing, the most plausible explanation is another \texttt{ping},
but speed changes or a stop command remain possible.
The protocol model is therefore essential to rank alternative explanations
and to estimate the current state.

In \emph{Probabilistic Trace Expressions} (PTEs) extend TEs by adding the definition of a probability distribution, associated with every event type $\theta$, depending on which other event types are feasible alternatives to  $\theta$. 
This approach supports both finite-state and infinite-state protocols
in a uniform and consistent manner.

%We refer to these respectively as probabilities due to \emph{observation} and to \emph{guessing}.
%Both perspectives are meaningful and will later be captured by two distinct semantics.

The PTE corresponding to $COGC$ is given by $AA$ plus the definition listed in Figure \ref{ex:probabilityAA}, which is coherent with the PFSA in Figure \ref{PFSA-fig}. W.r.t. our previous work, where probabilities were associated with transitions and hence needed a change in the TE syntax, this new approach to adding probabilities is incremental and modular. A PTE is just a plain TE, with no syntactic modifications, plus probabilities that are defined separately. 

\begin{figure*}
\begin{center}
\begin{math}
\begin{array}{c}
\eventSet = \{ping, stop, low, max, revert, keep\}
\\[3ex]

\Pi_{\{ping, low, max, stop\}}(ping) = 0.9; \Pi_{\{ping, low, max, stop\}}(low) = 0.02;
\\[3ex]

\Pi_{\{ping, low, max, stop\}}(max) = 0.05; \Pi_{\{ping, low, max, stop\}}(stop) = 0.03;
\\[3ex]

\Pi_{\{keep, revert, max, stop\}}(keep) = 0.15; \Pi_{\{keep, revert, max, stop\}}(revert) = 0.2; ...
\\[3ex]
\end{array}
\end{math}
\end{center}

\raggedright The subset $\{ping, low, max, stop\}$ uniquely identifies $AA$, $\{keep, revert, max, stop\}$ identifies $LS$, $\{keep, revert, low, stop\}$ identifies $MS$, and $\{keep, revert, low, max\}$ identifies $STA$. For example, $\Pi_{\{ping, low, max, stop\}}(ping) = 0.9$ means that the probability that $ping$ is observed, in a state where other possible observed actions are $low$, $max$, and $stop$, is $0.9$.
\caption{The PTE corresponding to $COGC$: the $AA$ TE plus probabilities.}\label{ex:probabilityAA}
\end{figure*}

\subsection{No More Lonely, Much Less Guesses}

After two years from the landing on Mars of the first \co, many more \co\ robots have been sent to Mars and wander on its surface. Communication among them is less expensive than communication with \gc\ and definitely faster. 

Let us suppose that when \gc\ wants to take the control over the \co\ operations, it broadcasts the same message $obey$ to all the \co s. If \co 1 observes a gap in a given time slot, it may ask the other \co s, or -- even better -- those which are closer to it, which message they received from \gc\ in the same time slot. By comparing their answers, \co 1 may make informed guesses on how to fill the gap. For example, if all the neighbouring \co s received messages different from $obey$, \co 1 may assume that the gap should not be filled by $obey$. But also in the  unfortunate case where other \co s observed a gap, which is indeed the case we will explore in Section \ref{impl-sec}, they can share the information on which messages they might have received in their current state, which in general should differ from agent to agent, along with their probabilities. 
By exchanging information on the sets of messages that could fill each \co's gap
%\footnote{The intersection operation makes sense for messages that are broadcast to all the agents, and for situations where there is some redundancy in the way events are observed, and by whom: the details will be clarified in Section \ref{impl-sec}.} 
and by managing probabilities, they can either reach a common understanding of the missing message or, at least, prune some of the open options.

%%%%%%%%%%%%%%%%%%%%%%%%%%%%%%%%%%%%%%%%%%%%%%%%%%%%%%%%%%%%%%%%%%%%%%%%%%%%%%%
\section{Probabilistic Trace Expressions}
\label{prob-trace-expr-sec}

As shown in the motivating scenario, the PTE syntax is the same as TE syntax, plus definition of probability. 

From a semantic point of view, two distinct probabilistic questions arise.
Given a sequence of observed messages and gaps, one may ask what is the probability
that the protocol is currently in a given state.
Alternatively, one may ask what is the probability that a specific sequence of messages
has occurred, including those that were not observed.

We refer to these respectively as probabilities due to \emph{observation} and to \emph{guessing}. Both perspectives are meaningful and are captured by two distinct semantics.

\subsection{Observation Semantics ({\em ObSem})}
\label{subsec:obsem}

A PTE represents the current protocol state after observing a sequence of events.
To track uncertainty, we define a \emph{PTE state} as a triple
$\globalState{\tau}{\pi}{obs}$,
where $\tau$ is a TE, $\pi$ is the probability of having reached $\tau$, and $obs$ is the observed trace.

Let $\tau_{init}$ be the initial TE. The initial PTE state is
$\globalState{\tau_{init}}{1}{\sigma}$, with $\sigma$ denoting the empty sequence of observed events.

Under {\em ObSem}, observing an event $\ev$ updates the state deterministically:
\[
\globalState{\tau}{\pi}{obs}
\transo{\ev}
\globalState{\tau'}{\pi}{obs~\ev}
\quad \text{if } \tau \trans{\ev} \tau'.
\]
Observed events have probability~1.

When a gap is observed, the protocol must advance, but the triggering event is unknown.
For each $\eventTy \in \poss(\tau)$:
\[
\globalState{\tau}{\pi}{obs}
\transo{gap}
\globalState{\tau'}{\Pi_{\poss(\tau)}(\eventTy)\cdot \pi}{obs~gap(\eventTy)},
\]
where $\tau \trans{\ev} \tau'$ and $\ev \in \eventTy$.
Thus, probability mass is distributed among all possible explanations for the gap.
Figure~\ref{trans-prob-states-fig_obs} summarises the semantics.
%\vspace{-0.5cm}
\begin{figure*}[!htb]
\begin{center}

\begin{math}
\begin{array}{c}
\Rule{obs-event}
{\tau \trans{\ev} \tau'}
{\globalState{\tau}{\pi_{tr}}{obs}
\transo{\ev}
\globalState{\tau'}{\pi_{tr}}{obs~\ev}  }
{}

\\[3ex]

\Rule{obs-gap}
{\tau \trans{\ev} \tau'}
{\globalState{\tau}{\pi_{tr}}{obs}
\transo{gap}
\globalState{\tau'}{\Pi_{\poss(\tau)}(\eventTy) * \pi_{tr}}{obs~gap(\eventTy)}  }
{\ev \in \eventTy}

\end{array}
\end{math}
\end{center}
\caption{Transition system for PTEs under {\em ObSem}.}\label{trans-prob-states-fig_obs}
\end{figure*}

Event types provide an abstraction that avoids enumerating all syntactic variants of events; gaps record event types rather than concrete events.

\subsection{Guess Semantics ({\em GuesSem})}

{\em GuesSem} propagates probabilities both for observed events and gaps.
Unlike {\em ObSem}, observed events are treated as probabilistic outcomes:
\[
\globalState{\tau}{\pi}{obs}
\transg{\ev}
\globalState{\tau'}{\Pi_{\poss(\tau)}(\eventTy)\cdot \pi}{obs~\ev},
\]
with an analogous rule for gaps.
This semantics corresponds to likelihood-based reasoning and aligns with HMM-style state estimation, as discussed in Section~\ref{hmmprop-subsect}.
Figure~\ref{trans-prob-states-fig_guess} presents Rule guess-event. Rule guess-gap is identical to Rule obs-gap and is not reported.
%\vspace{-0.5cm}
\begin{figure*}[!htb]
\begin{center}

\begin{math}
\begin{array}{c}
\Rule{guess-event}
{\tau \trans{\ev} \tau'}
{\globalState{\tau}{\pi_{tr}}{obs}
\transg{\ev}
\globalState{\tau'}{\Pi_{\poss(\tau)}(\eventTy) * \pi_{tr}}{obs~\ev}  }
{\ev \in \eventTy}
%
%\\[3ex]
%
%\Rule{guess-gap}
%
%{\tau \trans{\ev} \tau'}
%
%{\globalState{\tau}{\pi_{tr}}{obs}
%\transg{gap}
%\globalState{\tau'}{\Pi_{\poss(\tau)}(\eventTy) * \pi_{tr}}{obs~gap(\eventTy)}  }
%
%{\ev \in \eventTy}
\end{array}
\end{math}
\end{center}
\caption{Transition system for PTEs under {\em GuesSem}.}\label{trans-prob-states-fig_guess}
\end{figure*}
%%\vspace{-1cm}

\subsection{Nondeterminism and State Sets}

Even for deterministic TEs, gap transitions may yield multiple successor states.
We therefore lift PTE transitions from states to sets of states.
The functions $\transM{}$ and $\transS{}$ compute, respectively, the set of successor states from a single state and from a set of states, while $\transS{evs}$ denotes their transitive closure over a trace.
Figure~\ref{closure-fig} formalises these constructions.

\begin{figure*}[!htb]
\begin{center}

\begin{math}
\begin{array}{c}
\Rule{state-to-set}
{\gamma \trans{\anyg} \gamma_1~~~\gamma \trans{\anyg} \gamma_2~~~...~~~\gamma \trans{\anyg} \gamma_n}
{\gamma
\transM{\anyg}
\Gamma = \{ \gamma_1, \gamma_2, ..., \gamma_n \}}
{\{ \gamma_1, \gamma_2, ..., \gamma_n \} ~\textrm{are all and only the states reachable from}~ \gamma ~\textrm{via}~\trans{\anyg}}

\\[3ex]

\Rule{set-to-set}
{\gamma_1 \transM{\anyg} \Gamma_1~~~\gamma_2 \transM{\anyg} \Gamma_2~~~...~~~\gamma_n \transM{\anyg} \Gamma_n}
{\{ \gamma_1, \gamma_2, ..., \gamma_n \} \transS{\anyg} \Gamma_1 \cup \Gamma_2 \cup ... \cup \Gamma_n }
{}

\\[3ex]

\Rule{closure}
{\Gamma_0 \transS{\anyg_1} \Gamma_1 \transS{\anyg_2} ~~~...~~~ \Gamma_{n-1} \transS{\anyg_n} \Gamma_n }
{\Gamma_0 \transS{\anyg_1 ... \anyg_n} \Gamma_n}
{}

\\[3ex]

\Rule{closure-init}
{ \{ \globalState{\tau}{1}{\sigma} \} \transS{\anyg_1 ... \anyg_n} \Gamma_n }
{ \tau \transS{\anyg_1 ... \anyg_n} \Gamma_n}
{\sigma = {\textrm{empty sequence}}}
\end{array}
\end{math}

\end{center}
\caption{Rules for nondeterminism and transitive closure.}\label{closure-fig}
\end{figure*}

Under {\em ObSem}, probabilities reflect only uncertainty due to gaps; under {\em GuesSem}, they reflect both event generation and state reachability.
This distinction allows PTEs to support both belief revision from observations and likelihood-based trace analysis.

\section{The Theory: Exploring the Connections among PTEs, LTL, and HMM}
\label{sec:theory}

We investigate the theoretical connections between Probabilistic Trace Expressions (PTEs), Linear-Time Temporal Logic (LTL), and Hidden Markov Models (HMMs). After recalling the required background, we show how PTEs subsume HMM-style reasoning and support LTL satisfaction under partial observability.

\subsection{From HMMs to PTEs}
\label{hmmprop-subsect}

Manually defining the probability functions of a PTE can be error-prone. To mitigate this, we show how an existing HMM can be systematically translated into a PTE with equivalent generative behaviour, allowing learned or engineered probabilistic models to be reused.

The translation proceeds in two stages (HMM2PTE). First, the HMM is transformed into a Finite State Probabilistic Trace Expression (FSPTE), where each HMM state corresponds to a TE, and transition probabilities are encoded syntactically on prefixes. Since HMMs exhibit probabilistic branching that is observationally nondeterministic, event types causing ambiguity are renamed to preserve distinct probabilistic outcomes.

Second, the FSPTE is converted into a proper PTE by removing syntactic probabilities and defining the corresponding $\Pi_{\poss(\tau)}$ distributions. A default uniform distribution is used whenever needed.

This construction yields a PTE whose generative behaviour matches that of the original HMM.

\paragraph{Forward Algorithm Correspondence.}

The following result establishes the semantic correspondence between HMMs and PTEs under the {\em GuesSem} semantics.

\begin{theorem}\label{th:theo1}
Let $H$ be a HMM and let $\Gamma_0$ be the initial set of PTE states obtained via HMM2PTE, with probabilities given by the HMM initial distribution. For any observation sequence $\ev_1\ldots\ev_t$, the probability $\alpha_t(j)$ computed by the HMM forward algorithm equals the sum of the probabilities of all PTE states in $\Gamma_t$ whose TE corresponds to the HMM state $s_j$.
\end{theorem}
\begin{proof}
The correspondence follows from the fact that, under {\em GuesSem}, probability mass is propagated by multiplying the probability of the current state with the probability assigned to the observed event type, exactly as in the emission and transition steps of the HMM forward algorithm. The HMM2PTE construction ensures that each HMM state is represented by a corresponding TE, and that each observation sequence induces the same set of probabilistic paths in the PTE as in the original HMM. Summing the probabilities of all PTE states corresponding to a given HMM state therefore yields the same value computed by the forward algorithm.
\end{proof}

This theorem shows that {\em GuesSem} generalises classical HMM state estimation, embedding it within a more expressive trace-based framework that supports partial observability and non-regular specifications.

\subsection{LTL Satisfaction via PTEs}
\label{ltlprop-subsect}

To check whether a protocol modelled as a PTE satisfies an LTL property $\phi$, even in the presence of observation gaps, we translate $\phi$ into the same formalism.

The process consists of two steps. First, $\phi$ is translated into a non-probabilistic TE $\tau_{np}(\phi)$ using the algorithm in \cite{DBLP:conf/birthday/AnconaFM16}. Second, $\tau_{np}(\phi)$ is transformed into a probabilistic TE $\tau(\phi)$ by associating a suitable family of probability distributions (typically uniform).

Given a protocol PTE $AIP$, satisfaction of $\phi$ is checked by composing the two specifications via intersection:
\[
\tau(\phi) \andop AIP.
\]
A trace is compliant if it satisfies both the protocol constraints and the LTL property.

This approach integrates LTL specifications, PTE protocol models, and HMM-derived probabilistic behaviour within a single, uniform framework. It generalises existing RV approaches with state estimation \cite{DBLP:conf/rv/StollerBSGHSZ11} while supporting richer, non-regular specifications.

\section{The Practice: Minding Gaps}
%: Implementation and Issues}
\label{impl-sec}

In this section we move from theory to practice by discussing the practical issues raised by filling observation gaps, the solution we propose to address such issues, and the assumptions underlying such solution.

% \begin{figure}
%   \centering
%   \includegraphics[width=0.35\linewidth]{imgsAAA/totallyDec.jpg}
%   \caption{Fully decentralised monitoring.}\label{fig:totallyDec}
% \end{figure}
%\begin{figure}
%  \centering
%  \includegraphics[width=0.38\linewidth]{imgsAAA/totallyCen.jpg}
%  \caption{Runtime monitoring.}\label{fig:totallyCen}
%\end{figure}
% \begin{figure}
%   \centering
%   \includegraphics[width=0.45\linewidth]{imgsAAA/partiallyDec.jpg}
%   \caption{Partially decentralised monitoring.}\label{fig:partiallyDec}
% \end{figure}

The discussion so far has implicitly assumed that each agent is able to monitor its own communicative behaviour.
This simplifying assumption is useful to introduce the PTE formalism and its semantics, but it does not capture
the setting of interest in this work, namely the monitoring of a multi-agent system (MAS) from a global perspective.
In the scenario we consider, communicative events are observed and analysed by an external monitoring component.
%Figure~\ref{fig:totallyCen} depicts the 
In the fully centralised case, a single monitor $m$ is responsible for
observing and reasoning about all communicative events occurring in the MAS.

As far as events are concerned, they can be an online stream, with each event verified against the PTE specification as soon as it is observed  (\emph{online RV}), or can be recorded on a log file and inspected after the system shutdown, or once every some time (\emph{offline RV}). In both scenarios there may be gaps, due to different reasons. In offline RV, gaps might be caused by event sampling, as usually done to reduce the monitor workload. In online RV, a gap indicates lack of information (a lost message, event or perception); in this case, the absence of information is due to technical constraints of the system or of the monitor observation capabilities rather than to optimisation purposes.   

Differently from \cite{DBLP:conf/rv/StollerBSGHSZ11}, to perform RV using PTEs we need that each gap represents one single unobserved event: if we have a sequence of two unobserved events, we must have two different gaps in the observed trace. 
% If, in the real system, this ``one event-one gap'' correspondence cannot be achieved, we can estimate the number of unobserved events that took place in a time slot $T$ by computing the average rate of the event generation $G$, and inserting $T*G$ gaps in the event trace. As an example, if the monitor pauses for 3 seconds and the average events generation rate is 4 events for second, the trace will have 12 consecutive gaps corresponding to what might have happened in $T$.

The \emph{set-to-set} semantic rule (Figure \ref{closure-fig}) generates a set of states each time it is applied. These states must be maintained by the monitor in its local knowledge base, to allow it to retrieve the current set of states, query each of them, and  update the knowledge base with newly generated states.

Unfortunately, a rule like \emph{set-to-set} suffers from state space explosion, in particular when there are many sources of nondeterminism. Each time a gap takes place, the monitor must make guesses on the possible actual events that the gap represents and save all the states generated by these guesses. A possibly huge logical tree-like structure with states as nodes, and moves from states to states as edges, represents these open possibilities. If RV takes place online, the exploration of this logical structure must  follow a breadth-first strategy (more space needed but possibly less time required to recognise that the trace is not compliant with the expected behaviour), as the final trace of events is unknown and the levels of the structure are generated and explored at the same time. When, instead, a log file is analysed offline, the trace in the log is already complete and the logical tree-like structure can be explored, looking for violations, following a depth-first search (less space needed, but the violation could be discovered after exploring all the entire tree-like structure).

Online RV is definitely more challenging: if the log file is analysed offline, after the system has completed its execution, discovering a violation with some delay is not an issue. But if RV takes place online, it must be performed efficiently and in such a way that violations are discovered as soon as possible, to immediately take actions to repair or stop the system.

\subsection{Filling the Gaps}
\label{sec:centralised}

\begin{figure*}
Let us consider  a MAS involving four \co\ agents: $\{ co1, co2,$ $co3,$ $co4 \}$.
To make the presentation more readable we consider singleton event types denoting the message passing and we write $coi \transmsg{m_k} coj ~\prefixop~\tau$ instead of $\eventTy\prefixop\tau$ where $\eventTy = \{ coi \transmsg{m_k} coj \}$.
The PTE $\tau$ is defined as
\begin{center}
\begin{math}
\begin{array}{c}
\tau = \tau_1 \orop \tau_2
\\[3ex]

\tau_1 = co1 \transmsg{msg_1} co2 ~\prefixop~ (co2 \transmsg{msg_2} co3 ~\prefixop~ \tau_1 ~\shuffleop~ co2 \transmsg{msg_3} co4 ~\prefixop~ \epsilon)
\\[3ex]

\tau_2 = co1 \transmsg{msg_4} co4 ~\prefixop~ (co3 \transmsg{msg_5} co4 ~\prefixop~ \epsilon ~\shuffleop~ co2 \transmsg{msg_3} co4  ~\prefixop~ \tau_2)
\\[3ex]

\Pi_{\{ co1 \transmsg{msg_1} co2, co1 \transmsg{msg_4} co4 \}}(co1 \transmsg{msg_1} co2) = 0.7; \Pi_{\{ co1 \transmsg{msg_1} co2, co1 \transmsg{msg_4} co4 \}}(co1 \transmsg{msg_4} co4) = 0.3;
\\[3ex]

\Pi_{\{ co2 \transmsg{msg_2} co3, co2 \transmsg{msg_3} co4 \}}(co2 \transmsg{msg_2} co3) = 0.6; \Pi_{\{ co2 \transmsg{msg_2} co3, co2 \transmsg{msg_3} co4 \}}(co2 \transmsg{msg_3} co4) = 0.4;
\\[3ex]

\Pi_{\{ co3 \transmsg{msg_5} co4, co2 \transmsg{msg_3} co4 \}}(co3 \transmsg{msg_5} co4) = 0.3; \Pi_{\{ co3 \transmsg{msg_5} co4, co2 \transmsg{msg_3} co4 \}}(co2 \transmsg{msg_3} co4) = 0.7
\\[3ex]

\Pi_{\ET S} = \Pi u_{\ET S}~\textrm{for any other}~\ET S
\end{array}
\end{math}
\end{center}

\noindent and the initial probability of $\tau$ is $1$.
\caption{Filling gaps in a centralised way (one monitor).}\label{ex:decentralisePTEex}
\end{figure*}

Let us consider the scenario illustrated in Figure \ref{ex:decentralisePTEex}. A monitor $M$  observing all the interactions among the agents starting from the state $\tau$ would behave in the following way. We identify the initial state of $M$ with $M_{0}$, s.t. $M_{0} = \globalState{\tau}{1}{\sigma}.$
Let us suppose that the first observed event is a gap. Starting from $\tau$,  the only two possible evolutions of the protocol are those where either $co1 \transmsg{msg_1} co2$ or $co1 \transmsg{msg_4} co4$. These evolutions may be formalised as

$$
M_{0} \transS{gap} M_{1} = \{ 
$$
$$
\globalState{co2 \transmsg{msg_2} co3 ~\prefixop~ \tau_1 ~\shuffleop~ co2 \transmsg{msg_3} co4 ~\prefixop~ \epsilon}{0.7}{\gap{co1\transmsg{msg_1}co2}}, 
$$
$$
\globalState{co3 \transmsg{msg_5} co4 ~\prefixop~ \epsilon ~\shuffleop~ co2 \transmsg{msg_3} co4  ~\prefixop~ \tau_2}{0.3}{\gap{co1\transmsg{msg_4}co4}}
$$
$$
\} 
$$

\noindent We omit the $o$ or $g$ subscript in $\transS{gap}$ since when gaps are observed, {\em ObSem} and {\em GuesSem} behave in the same way: $\transS{gap}$  could be either $\transS{gap}_o$ or $\transS{gap}_g$ and nothing would change. If another gap is observed, each state in $M_{1}$ can evolve in two different ways because of the shuffle, leading to

$$
M_{1} \transS{gap} M_{2} = \{ 
$$
$$
\globalState{\tau_1 ~\shuffleop~ co2 \transmsg{msg_3} co4 ~\prefixop~ \epsilon}{0.42}{\gap{co1\transmsg{msg_1}co2}~\gap{co2 \transmsg{msg_2} co3}}, 
$$
$$
\globalState{co2 \transmsg{msg_2} co3 ~\prefixop~ \tau_1}{0.28}{\gap{co1\transmsg{msg_1}co2}~\gap{co2 \transmsg{msg_3} co4}}, 
$$
$$
\globalState{co2 \transmsg{msg_3} co4  ~\prefixop~ \tau_2}{0.09}{\gap{co1\transmsg{msg_4}co4}~\gap{co3 \transmsg{msg_5} co4}}
$$
$$
\globalState{co3 \transmsg{msg_5} co4 ~\prefixop~ \epsilon ~\shuffleop~ \tau_2}{0.21}{\gap{co1\transmsg{msg_4}co4}~\gap{co2 \transmsg{msg_3} co4}}
$$
$$
\} 
$$

\noindent However, if rather than a gap, the event $co2 \transmsg{msg_2} co3$ is observed, then $M_1$ evolves back to a single term (a singleton set), leading to
$$
M_{1} \transS{co2 \transmsg{msg_2} co3} M_{2}' = \{ 
$$
$$
\globalState{\tau_1 ~\shuffleop~ co2 \transmsg{msg_3} co4 ~\prefixop~ \epsilon}{\pi}{\gap{co1\transmsg{msg_1}co2}~co2 \transmsg{msg_2} co3}, 
$$
$$
\} 
$$
with $\pi=0.42$, if \emph{GuesSem} is selected, while $\pi=0.7$, if \emph{ObSem} is. 

Note that, even though the second observed event allows us to discard the second term in $M_1$, the probability assigned to observing $co1 \transmsg{msg_1} co2$ as first event remains $0.7$. It cannot become $1.0$ because it is possible the first event was actually $co1 \transmsg{msg_4} co4$, and the second event a violation of the specification. 

It is easy to see that the number of states can rapidly grow, because one single monitor is in charge for the RV of all the MAS and takes care of all the possibilities that open up when gaps are observed. One approach to cope with state space explosion is to split the centralised monitor into a set of decentralised ones, each observing a portion of the MAS. 
%If DecAMon is used to compute such ``portions'', monitoring safety is guaranteed. Since each decentralised monitor has to make its guesses about gaps, when a gap is observed there may be different opinions about its possible values. With respect to a centralised approach,  different perspectives due to decentralisation need to be managed through synchronisation between the monitors, which generates some communication overhead. The experimental results presented in Section \ref{experiments-sec} show that the overhead cost is worth paying. 

% \subsection{From Centralised to Decentralised Gaps Management}
% \label{sec:decentralised}

The most challenging situation for a decentralised monitoring process  is the one where all the monitors ``observe a gap'', as they cannot help each other by sharing information on what they saw. Other situations are a special case of this one, and less challenging: if some monitors observe events instead of gaps, they can provide valuable information to other monitors to reduce the state space due to guesses associated with gaps.

\subsection{Implementation Details}
\label{sec:impldetails}

Proof of concepts of the algorithms presented in this section have been implemented using SWI-Prolog\footnote{\url{http://swi-prolog.org/}, accessed on \today.}. The code is available from \url{https://github.com/RMLatDIBRIS/ProbabilisticTraceExpressions} under GNU General Public License v3.0\footnote{\url{https://www.gnu.org/licenses/gpl-3.0.en.html}, accessed on \today.}.

PTEs can be easily modelled as Prolog terms; by exploiting syntactic equations where the same logical variable appears both to the left and to the right of the ``='' syntactic equality symbol,  recursive PTEs can be defined. This feature is supported by most Prolog implementations, including SWI-Prolog,  and allows us to define the PTEs shown in the examples provided so far with almost the same syntax used in the paper. Implementing $\Pi$ in Prolog is not an issue, as it is represented as a set of facts. 

The adoption of Prolog is a winning choice not only for representing PTEs, but also for implementing their semantics and for manipulating them. 
Thanks to Prolog's rule-based, declarative interpretation, the rules defining PTE  operational semantics have a one-to-one correspondence with Prolog clauses: backtracking and ``all-solutions'' predicates are powerful tools to deal with the generation of multiple PTE states, when gaps or intrinsic nondeterminism coming from HMMs transformed into PTEs make more states reachable from a given one   (\emph{set-to-set} rule). 
%Finally, as we already anticipated, the SWI-Prolog coinduction library allows us to manipulate cyclic terms avoiding loops in the execution. 

A SWI-Prolog PTE-driven monitor observing events taking place in the system under scrutiny, and checking whether they comply with the PTE or not, can be automatically generated from the PTE representation. 
Connectors with such SWI-Prolog PTE-driven monitors exist both for MASs \cite{DBLP:conf/idc/BriolaMA14} and for other systems, including IoT \cite{DBLP:journals/corr/abs-1802-01790,DBLP:conf/icwe/LeottaAFORR18} and object oriented applications \cite{DBLP:conf/ecoop/AnconaFFM17}. 

\subsection{Thoughts on the Origins of Gaps}

Before moving on with the conclusions, it is important to linger a bit on what kind of scenarios can bring to having gaps in the trace of events.

In Section~\ref{motiv-sec}, we presented a motivating scenario involving the Mars \co. In such scenario, the need of probabilistic specifications was determined by the possible lack of information deriving from the disruption of the signal, delays in the communication, and so on. For these reasons, the event traces to be analysed at runtime contained gaps, which represented the lack of information. Even though this was a key example of why gaps can be present in a trace of events, it is not the only one. Along with the lack of information, other reasons can be found for having gaps in the trace of events. Specifically, gaps can derive by uncertainty over the observed events. This can happen in scenarios where the trace of events has been obtained by combining multiple traces together, with each trace representing a different viewpoint of the same scenario. For instance, the system execution has been observed from different perspectives, and there are points of conflict, where two (or more) observers disagree on what event has been observed in a certain point of the trace. When this happens, we find ourselves in a case of uncertainty; since we do not know which of the observers is right (if any). Such uncertainty is not that different from the lack of information we showed for \co, since there are points in the trace where no certainty of which event has been observed are present. In these scenarios, PTE can be of help, by handling the uncertainty as gaps of information.

Let us make an abstract example, where we have three different event traces, generated by three different observers of the same scenario. The three traces are:
$$
ev_1~ev_2~ev_3~ev_4;~~~ 
ev_1~ev_5~ev_3~ev_4;~~~
ev_1~ev_2~ev_3~ev_6
$$

Now, as we can see, the three traces (named ``multi-trace'', that could be partially observed as in \cite{DBLP:journals/jot/MaheBGLG24}, or just inconsistent as above) differ from each other for some event. Specifically, the first one differs from the second one w.r.t. the second event in the trace, the first one differs from the third one w.r.t. the fourth event, and finally, the second one differs from the third one w.r.t. the second and fourth event.
We can combine these three traces to generate a single trace with gaps, as follows:
$$
ev_1~gap~ev_3~gap
$$
where the certain events are preserved (the ones all observers agree on), and the uncertain events are replaced with gaps (the ones at least two observers disagree on).

This approach, that we proposed and implemented - via a proof of concept - in 2021 \cite{DBLP:journals/scp/MaheBGG25}, is similar to the Partial Order Reduction technique proposed by and consisting in selecting a one-unambiguous action (as a unique first step to a linearization) by confronting multi-traces with interaction specifications.

\section{Conclusions}
\label{concl-sec}

This paper addressed the problem of runtime verification of multi-agent interaction protocols under partial observability.
We proposed \emph{Probabilistic Trace Expressions} (PTEs) as a uniform modelling and monitoring framework that integrates probabilistic reasoning into Trace Expressions, allowing uncertainty due to missing or unobservable events to be handled explicitly, in 2022  \cite{DBLP:conf/eumas/AnconaFM22,DBLP:conf/cilc/AnconaFM22}.

In this paper, we keep the main ideas of PTEs but greatly simplify their syntax, and we propose two different semantics for them.

The main strength of PTEs lies in their expressive power combined with principled probabilistic semantics.
Unlike most approaches in the RV literature, which rely on LTL, finite-state automata, or Hidden Markov Models, PTEs can naturally capture non-regular interaction patterns through operators such as concatenation and shuffle, while still subsuming classical probabilistic models.
We showed that HMMs and LTL specifications can be systematically translated into PTEs, whereas the converse is not possible due to the limited expressive power of those formalisms.

Moving beyond the RV literature, we observed that similar challenges related to partial observability and uncertainty arise in the monitoring of normative multi-agent systems.
Although norms and PTEs originate from different research traditions, the overlap in the underlying problems suggests that a deeper comparison between these formalisms is a promising direction for future work.

Planned extensions include the introduction of parameters into PTEs, the formal analysis of the computational complexity of the proposed monitoring algorithms, and experimentation in real-world settings.
Finally, building on the integration of probabilities into Trace Expressions, we aim to extend the Runtime Monitoring Language (RML)~\cite{DBLP:journals/scp/AnconaFFM21} with a probabilistic variant compiling into PTEs, further supporting expressive and uncertainty-aware runtime verification.

\printbibliography

\end{document}